\documentclass[11pt, letter]{article}
\usepackage{fullpage}
\usepackage{amsmath,amsfonts,amsthm,amssymb,mathrsfs}
\usepackage{mathtools}
\usepackage{setspace}
\usepackage{graphicx,float,wrapfig}
\usepackage{hyperref} 
\usepackage{caption}
\usepackage{tikz}
\usetikzlibrary{arrows,decorations.pathmorphing,backgrounds,positioning,fit,petri}

\theoremstyle{plain} \newtheorem{theorem}{Theorem}[section]  \newtheorem{corollary}[theorem]{Corollary}  \newtheorem{observation}[theorem]{Observation} 
\theoremstyle{definition}\newtheorem{definition}[theorem]{Definition} 
\newtheorem*{noncorollary}{Corollary}
\newtheorem{innercustomthm}{Invariant}
\newenvironment{Invariant}[1]
  {\renewcommand\theinnercustomthm{#1}\innercustomthm}
  {\endinnercustomthm}

\title{A Faster External Memory Priority Queue with DecreaseKeys}
\author{Shunhua Jiang \thanks{Institute for Interdisciplinary Information Sciences, Tsinghua University; jiangsh15@mails.tsinghua.edu.cn. Work done while staying at Aarhus University.}
\and Kasper Green Larsen \thanks{Department of Computer Science, Aarhus University; larsen@cs.au.dk.  Supported by a Villum Young Investigator Grant and an AUFF Starting
Grant.}}
\date{}

\begin{document}
\maketitle

\begin{abstract}
A priority queue is a fundamental data structure that maintains a dynamic set of (key, priority)-pairs and supports Insert, Delete, ExtractMin and DecreaseKey operations. In the external memory model, the current best priority queue supports each operation in amortized $O(\frac{1}{B}\log \frac{N}{B})$ I/Os. If the DecreaseKey operation does not need to be supported, one can design a more efficient data structure that supports the Insert, Delete and ExtractMin operations in $O(\frac{1}{B}\log \frac{N}{B}/ \log \frac{M}{B})$ I/Os. A recent result shows that a degradation in performance is inevitable by proving a lower bound of $\Omega(\frac{1}{B}\log B/\log\log N)$ I/Os for priority queues with DecreaseKeys. In this paper we tighten the gap between the lower bound and the upper bound by proposing a new priority queue which supports the DecreaseKey operation and has an expected amortized I/O complexity of $O(\frac{1}{B}\log \frac{N}{B}/\log\log N)$. Our result improves the external memory priority queue with DecreaseKeys for the first time in over a decade, and also gives the fastest external memory single source shortest path algorithm.
\end{abstract}

\section{Introduction}
\noindent The priority queue is a fundamental data structure that has many applications. For example, a priority queue can be used in Prim's algorithm to find a minimum spanning tree and it can be used in Dijkstra's algorithm to find single source shortest paths. A priority queue maintains a dynamic set $S$ of \emph{entries} of the form (key, priority). It supports the following basic operations:
\begin{itemize}
\item Insert$(k,p)$: Insert an entry with key $k$ and priority $p$ into $S$. 
\item Delete$(k,p)$: If the entry $(k,p)$ is present in $S$, delete the entry.
\item ExtractMin$()$: Remove and return the entry with minimum priority from $S$.
\item DecreaseKey$(k,p)$: Given key $k$ where there already exists an entry $(k,p')$ in $S$, the priority associated with $k$ is changed to $p$ if and only if $p<p'$. 
\end{itemize}
The priority queue is a well studied data structure. Let $N$ be the maximum size of $S$. The classic Fibonacci heap \cite{fredman1987fibonacci} can support the Insert operation in worst case $O(1)$ time, ExtractMin in amortized $O(\log N)$ time and DecreaseKey in amortized $O(1)$ time. When the priorities are only comparable, this time complexity is optimal because a priority queue can be used to implement comparison-based sorting which has a lower bound of $\Omega(N\log N)$. If we assume priorities are integers, a priority queue that uses the reduction to integer sorting \cite{thorup2007equivalence} can support each operation in deterministic amortized $O(\log\log N)$ time \cite{MR2121186} or expected amortized $O(\sqrt{\log\log N})$ time \cite{han2002integer}. 

\paragraph{External Memory Model.} 
In modern computers and database systems, often there is a small internal memory that we can read from and write to very quickly and a large external memory that requires much longer time to access. 
Usually the data is too large to fit into the internal memory. The random access to the external memory takes much longer time than CPU instructions and the performance of an algorithm mainly depends on the number of such random accesses it makes.

The external memory model introduced by Aggarwal and Vitter \cite{aggarwal1988input} captures the fact that random accesses to the external memory is the bottleneck of an algorithm.
In this model, there is a small main memory of $M$ words and a large disk which is divided into blocks of $B$ words. Each word has size $w=\Theta(\log N)$ bits. An I/O can transfer one block between the disk and the memory. The performance of an algorithm is measured in terms of the total number of I/Os it makes. In this model, the tight lower bound for comparison-based sorting is $\Omega(\frac{N}{B}\log \frac{N}{B}/ \log \frac{M}{B})$ I/Os \cite{aggarwal1988input}.

\paragraph{Priority Queues in External Memory.}
The priority queue is also an important data structure in the external memory model. In \cite{kumar1996improved} Kumar and Schwabe proposed an external memory tournament tree that supports the Insert, Delete, ExtractMin and DecreaseKey operations of a priority queue in amortized $O(\frac{1}{B}\log \frac{N}{B})$ I/Os, assuming that the keys are integers in $\{1, \dots, N\}$. This priority queue can be used in the external memory Dijkstra's algorithm also described in \cite{kumar1996improved} for Single Source Shortest Paths (SSSP). When running the algorithm on a graph $G(V,E)$, the vertices are labeled $1,\dots, |V|$ and the priority queue takes these labels as keys. The algorithm performs $O(|E|)$ ExtractMins and $O(|E|)$ DecreaseKeys on a priority queue and it costs $O(|V|)$ extra I/Os. In total it costs $O(|V|+\frac{|E|}{B}\log \frac{|V|}{B})$ I/Os using the external memory tournament tree. Hence an improvement to the external memory priority queue also leads to an improvement to SSSP in the external memory model for general graphs. 

It is also known that if DecreaseKeys do not need to be supported, there exist more efficient priority queues that run in amortized $O(\frac{1}{B}\log \frac{N}{B}/ \log \frac{M}{B})$ I/Os \cite{MR1984614}\cite{MR1693788}.
Note that a DecreaseKey operation does not know the original priority of the key, so it does not know if the new priority is smaller, hence it cannot be simulated straightforwardly by a Delete operation that deletes the original (key, priority)-pair and an Insert operation that inserts the new pair. This feature is crucial in some applications, e.g., when performing DecreaseKeys in the external memory Dijkstra's algorithm, there is no I/O-efficient way to know the original priority of a vertex.

The external memory tournament tree is a $\log \frac{M}{B}$ factor slower than priority queues without DecreaseKeys. A natural question is whether this extra cost is really necessary. Recently, this question is answered affirmatively in \cite{MR3678253} where Eenberg, Larsen and Yu proved a lower bound of amortized $\Omega(\frac{1}{B}\log B/\log\log N)$ I/Os per operation for priority queues with DecreaseKeys. When $B\geq N^{\varepsilon}$ and under the common tall cache assumption $M\geq B^{1+\varepsilon}$, the lower bound becomes $\Omega(\frac{1}{B}\log N/\log\log N)$, whereas the I/O complexity of priority queues without DecreaseKeys becomes $O(\frac{1}{B}\log \frac{N}{B}/ \log \frac{M}{B})=O(\frac{1}{B})$ for this setting of parameters. Thus there is a $\log N/\log\log N$ gap between the performance of priority queues with and without DecreaseKeys. This shows that indeed more I/Os are required to support DecreaseKeys. The lower bound allows Las Vegas randomization and it also holds for non-comparison-based priority queues. Moreover, it allows arbitrary computation on the bits of priorities and keys such that they do not need to be treated as atomic entities.

\paragraph{Our Result.}
There still exists a gap between the lower bound of $\Omega(\frac{1}{B}\log B/\log\log N)$ in \cite{MR3678253} and the performance of $O(\frac{1}{B}\log \frac{N}{B})$ in \cite{kumar1996improved} for priority queues with DecreaseKeys. In this paper we tighten the gap by proposing a better data structure:
\begin{theorem}\label{thm:main}
When the keys are integers in $\{1, \dots, N\}$, and the priorities are comparable and can be stored in one word of size $w=\Theta(\log N)$ bits, under the assumption that $M>B\cdot \log^{0.01}N$, there exists a priority queue that supports the Insert, Delete, ExtractMin and DecreaseKey operations in expected amortized $O(\frac{1}{B}\log \frac{N}{B}/\log\log N)$ I/Os.
\end{theorem}

There has been no improvement to the external memory priority queue for over a decade until our result. For the natural setting of parameters where $B=N^{\varepsilon}$, our data structure is in fact optimal as it matches the lower bound of $\Omega(\frac{1}{B}\log B/\log\log N)$. A direct application of our data structure is to improve the I/O-complexity of external memory Dijkstra's algorithm to $O(|V|+\frac{|E|}{B}\log \frac{|V|}{B}/\log\log |V|)$ I/Os, which is the current best performance for SSSP in general~graphs. 

Also note that the assumptions that the keys are integers in $\{1,\dots, N\}$ and that the priorities can be stored in one word are also made in \cite{kumar1996improved}. The further assumption $M>B\cdot \log^{0.01}N$ is reasonable because if $\log\frac{M}{B} = o(\log\log N)$, our desired I/O complexity $O(\frac{1}{B}\log \frac{N}{B}/\log\log N)$ would surpass the lower bound of $\Omega(\frac{1}{B} \log \frac{N}{B}/\log \frac{M}{B})$ for comparison-based sorting in the external memory model. The number 0.01 in this assumption is set arbitrarily and can be replaced by any constant smaller than 1. Our solution makes use of hashing on the keys to create small summaries of the entries stored in a block on disk, but otherwise treats entries as atomic entities.

\paragraph{Related Work.}
The equivalence between sorting and priority queues in the external memory model is proved in \cite{wei2014equivalence}, which is analogous to the result in internal memory \cite{thorup2007equivalence}. However, the reduction from sorting algorithms to external memory priority queues in \cite{wei2014equivalence} only yields priority queues without DecreaseKeys. The priority queues without DecreaseKeys can also be made to run in worst case such that $B$ consecutive operations take $O(\log \frac{N}{B}/ \log \frac{M}{B})$ I/Os \cite{brodal1998worst}.

There exist extensive results both for data structures and lower bounds in the external memory model. One of the most well studied problem is the dictionary problem, where a dynamic set is maintained to support insertions, deletions and membership queries. The traditional solutions are the B-tree 
\cite{bayer1970organization} and the buffer tree \cite{MR1984614}. Many variants of the B-tree were also proposed to solve variants of the dictionary problem \cite{arge2003optimal}\cite{becker1996asymptotically}. A tight lower bound and an optimal data structure for the dictionary problem with integers is proposed in \cite{iacono2012using}. 

For Single Source Shortest Paths in sparse graphs, there is a faster algorithm~\cite{meyer2006efficient} which runs in $O(\sqrt{|E||V|/B}\log |V|)$ I/Os. This complexity is better than the $O(|V|+\frac{|E|}{B}\log \frac{|V|}{B})$ I/Os of external memory Dijkstra's algorithm when the graph is sparse and the $O(|V|)$ term becomes the dominant term. 
Similar to Dijkstra's algorithm, many other graph algorithms also have efficient analogues in the external memory model, e.g., BFS, DFS and minimum spanning tree  \cite{MR1321845}\cite{kumar1996improved}. These graph algorithms usually exploit the inherent parallelism of the external memory model and have a factor of $\frac{1}{B}$ in their I/O complexity.

A common variant of the external memory model is the cache-oblivious model proposed by Frigo et al. \cite{frigo1999cache}. In the cache-oblivious model, the parameters $M$ and $B$ are not explicitly known and the analysis of algorithms and data structures should hold for any $M$ and $B$. An advantage of the cache-oblivious model is that the analysis can be applied to all levels of a multi-level memory hierarchy. Analogous data structures are proposed for the cache-oblivious model, including search trees \cite{bender2000cache}\cite{brodal2002cache} and priority queues. In the cache-oblivious model the priority queues with \cite{chowdhury2004cache}\cite{brodal2004cache} or without DecreaseKeys \cite{MR2121151}\cite{brodai2002funnel} can both achieve the same I/O complexity as in the external memory model. In this paper we did not attempt to extend our priority queue to the cache-oblivious model.

\section{Data Structure}\label{sec:data}
\subsection{Previous Solution}
Our data structure is based on the external memory tournament tree in \cite{kumar1996improved}. It is convenient to give a brief introduction of this data structure first. It supports the following basic operations:
\begin{itemize}
\item Delete$(k)$:
If there is an entry with key $k$ stored in the data structure, delete the entry.
\item ExtractMin$()$:
Remove and return the entry with minimum priority from the data structure.
\item Update$(k, p)$:
If there already is an entry $(k, p')$ stored in the data structure, update it to $(k,p)$ if $p < p'$. If no entry with key $k$ is stored in the data structure, insert an entry $(k,p)$ into the data structure.
\end{itemize}
DecreaseKey and Insert operations can both be implemented using the Update operations. Also note that the data structure supports a stronger version of the Delete operation than the one described in the introduction, i.e., the Delete operation above does not need to know the priority of the entry being deleted.

The external memory tournament tree is a binary tree with $\frac{N}{M}$ leaves. Each node in the tree stores $\frac{M}{2}$ to $M$ entries. The leaves are numbered from 1 to $\frac{N}{M}$. Keys in the range $[(i-1)M+1, iM]$ are mapped to the $i$th leaf. Each entry is stored either in its corresponding leaf or in an ancestor of that leaf. 
The root is stored in memory and all the other nodes are stored on disk.

When performing an operation, a corresponding signal is first sent to the root. When a signal reaches a node, if the node contains the target entry of the operation, the correct changes are applied to the entries in that node, otherwise the signal should be propagated down to one of its children. To achieve amortized $O(\frac{1}{B}\log \frac{N}{B})$ I/Os for each operation, a buffering method is used: Each internal node has a signal buffer of size $M$ words where the signals to be propagated down are first stored in. When the signal buffer has accumulated $M$ signals, these signals are pushed down to the two children together. It should be noted that with the signal buffer, the ordering of the signals are still maintained. A signal that comes first is always stored in a lower node in the tree or in a position closer to the front of a signal buffer. Hence the signal buffers only delay the propagation of signals but do not affect the final result when the signals are all applied.

When pushing down the signals, the two children are loaded into memory and we check if a signal should be applied to the entries of the children. If not, it is put into the signal buffer of one of the children. In total this procedure costs $O(\frac{M}{B})$ I/Os if we do not consider the recursive calls to the procedures of pushing down signals at lower levels. These I/Os result in an amortized cost of $O(\frac{1}{B})$ I/Os for every signal that is involved. Since each signal can only go into one signal buffer at each level, the total amortized cost for each operation is $h\cdot \frac{1}{B}=O(\frac{1}{B}\log \frac{N}{B})$ I/Os where $h$ is the height of the tree. 

\subsection{Our Data Structure} \label{sec:property}
A natural idea to improve the external memory tournament tree in \cite{kumar1996improved} is to increase the degree of the tree to $t$ and thus decrease the height of the tree and the amortized I/O complexity by a $\log t$ factor. But note that after this change, loading all the children of a node into memory would cost $O(t\cdot \frac{M}{B})$ I/Os and thus pushing down the signals in a signal buffer would also cost $O(t\cdot \frac{M}{B})$ I/Os. Let $h =O(\log \frac{N}{B}/\log t)$ be the height of the new tree. The amortized cost becomes $h\cdot \frac{t}{B}=O(\frac{t}{B}\log \frac{N}{B}/\log t)$ I/Os and no improvement is achieved. To make this approach work, we store a small summary of the entries in each node so that we only need to read these small summaries into memory to decide which child a signal should be applied to. 

More concretely, our basic data structure is a static $t$-ary tree where $t = \log^{0.01}N$. Another difference is that in our data structure, each node stores $O(tB)$ entries instead of $O(M)$ entries.
 The height of the tree is $h=\log_{t}\frac{N}{tB} = O(\log\frac{N}{B}/\log\log N)$. We store a \emph{Bloom filter replacement}~\cite{MR2298337} in each node to record the keys stored in that node. A Bloom filter replacement supports both insertions and deletions and it can only make false positive errors, i.e., some keys may appear to be present in the Bloom filter replacement but actually they are not. In our data structure each Bloom filter replacement makes a false positive error with probability at most $\varepsilon = \frac{1}{\log^3N}$ and it is able to store up to $2tB$ keys. The results of \cite{MR2298337} ensures that such a Bloom filter replacement can be stored in $O(B)$ words. We delay a detailed description of the Bloom filter replacement to Section \ref{sec:bfr}. When pushing the signals in the signal buffer of a node down to its children, we don't need to read into memory all of the $O(t^2B)$ entries stored in the list of the children, instead, we only need to read into memory the $O(tB)$ words of all the Bloom filter replacements of the children. 

Since we do not load the actual entries of a node into memory, we cannot apply the signals to the entries of the children immediately, so we also store an extra \emph{todo buffer} of size $O(B)$ words in each node. When pushing the signals from a parent node down to its children, the signals that should be applied to a child are first stored in the todo buffer of that child. When the todo buffer has accumulated $B$ signals, these signals are applied to the entries together.

Each node of the tree stores several components with the following properties (see Figure \ref{fig}):
\begin{enumerate}
\item Each node of the tree has a list that can store $0$ to $2tB$ entries with different keys. The entries in the list are sorted from start to end in descending order according to priority.
\item The tree has $\lceil \frac{N}{tB}\rceil$ leaves. Denote the leaves as $l_1,l_2,\dots,l_{\lceil \frac{N}{tB}\rceil}$ from left to right. We define a function Leaf$(k)$ that maps a key $k$ in the range $\left[\left(i-1\right)tB+1,\, itB\right]$ to the $i$th leaf $l_i$. Each entry with key $k$ can only be stored in Leaf$(k)$, or in an ancestor of Leaf$(k)$.
\item Each internal node has a signal buffer with maximum $tB$ entries, and each node except the root has a todo buffer with maximum $B$ entries. The signal buffer stores signals to be pushed down to the children of the node and the todo buffer stores signals to be applied to the entries in the list of that node. The signals that come later are always appended at the end of a buffer. Each node also stores a Bloom filter replacement of size $O(B)$ words that records which keys are present in the list of the node with a false positive error rate of at most $\varepsilon = \frac{1}{\log^3N}$. Overall the size of each node is $O(tB)$ words.
\item Each node stores a boundary value which upper bounds the priorities of the entries in the list of that node (but is not necessarily the maximum priority). We use Boundary$(v)$ to denote the boundary value of node $v$. Initially all boundary values are set to $+\infty$. 
\item The root of the tree is stored in memory. Other nodes are stored on disk.
\end{enumerate}

\section{Invariants}
Our data structure maintains some invariants that we will use to prove the correctness of our data structure later in Section \ref{sec:correctness}. 
These invariants are maintained after each operation. It is convenient to define the invariants before introducing the operations because some details of the operations are designed to maintain the invariants.
We first define the following useful notions. 
\begin{definition}[Time Order]\label{def:time_order}
We define a \emph{time order} for all the signals and entries with the same key in the data structure. Intuitively the signals and entries in lower nodes are put into the data structure first and they should be processed first as well. The time order is defined by the following traversal of signals and entries: For a key $k$, we start from Leaf$(k)$ and move upward to the root. Inside each node, the signal buffer is first traversed, then the list and finally the todo buffer. Each buffer or list is traversed from start to end. If a signal/entry $\sigma_1$ is encountered before another signal/entry $\sigma_2$, we say that $\sigma_1$ precedes $\sigma_2$ in the time order. Figure \ref{fig} shows an illustration.
\end{definition}
\begin{definition}[Actual Priority]
We use Actual$(k,v)$ to denote the \emph{actual priority} of a key $k$ in a node $v$, which is the priority of $k$ in the list of $v$ after all signals in the todo buffer of $v$ are applied.

More formally, Actual$(k,v)$ is defined by the following algorithm: Initially Actual$(k,v)$ is set to $p$ if the list of $v$ contains an entry $(k,p)$ and to $+\infty$ if not. Then we go through the signals in the todo buffer of $v$ from start to end, i.e., follow the time order. If we encounter a Delete$(k)$ signal, Actual$(k,v)$ is updated to $+\infty$, and if we encounter an Update$(k,p')$ signal, Actual$(k,v)$ is updated to $p'$ if $p'$ is smaller than the current value of Actual$(k,v)$.

If Actual$(k,v)$ is finite, we also call the entry $(k,$ Actual$(k,v))$ an \emph{actual entry} stored in node $v$.
\end{definition}

\begin{definition}[Final Priority]
We use Final$(k,v)$ to denote the \emph{final priority} of a key $k$ up to a node $v$, which is the priority of $k$ after applying all signals of $k$ following the time order from Leaf$(k)$ up to $v$.

Formally, Final$(k,v)$ is defined by the following algorithm: If $v$ is a leaf, Final$(k,v)=$ Actual$(k,v)$. If $v$ is an internal node, initially Final$(k,v)$ is set to Final$(k,c)$ where $c$ is the child of $v$ that is on the path from Leaf$(k)$ to $v$. Then we go through the signals in the signal buffer of $v$ from start to end. If we encounter a Delete$(k)$ signal, Final$(k,v)$ is updated to $+\infty$, and if we encounter an Update$(k,p')$ signal, Final$(k,v)$ is updated to $p'$ if $p'$ is smaller than the current value of Final$(k,v)$. Finally, Final$(k,v)$ is set to Actual$(k,u)$ if Actual$(k,u)$ is smaller than the current value of Final$(k,v)$.
\end{definition}

\begin{figure}[t]
\centering
\begin{tikzpicture}
[treenode/.style={rectangle,inner sep=0pt,minimum height=3cm,minimum width=4cm},
rootnode/.style={rectangle,inner sep=0pt,minimum height=2cm,minimum width=4cm, label={\small \itshape root}},
list/.style={rectangle,inner sep=0pt,minimum height=0.3cm,minimum width=3.2cm,label={[xshift=-1.56cm, label distance=-2pt]\small \itshape list}},
todobuffer/.style={rectangle,inner sep=0pt,minimum height=0.3cm,minimum width=0.8cm, label={[xshift=0.22cm, label distance=-2pt]\small \itshape todo buffer}},
bloom/.style={rectangle,inner sep=0pt,minimum height=0.3cm,minimum width=0.8cm, label={[xshift=0.27cm, label distance=-2pt]\small \itshape other parts}},
signalbuffer/.style={rectangle,inner sep=0pt,minimum height=0.3cm,minimum width=1.6cm,label={[xshift=-0.06cm, label distance=-2pt]\small \itshape signal buffer}},
value/.style={rectangle,inner sep=0pt,minimum height=0.3cm,minimum width=0.4cm}],

\node (memory) at (5.6cm, 3.6cm) [thick, dotted, inner sep=0pt,minimum height=2.9cm,minimum width=16cm, label={[xshift=-7cm,yshift=-0.6cm]\small \itshape memory}, draw] {};
\node (disk) at (5.6cm, -0.05cm) [thick, dotted, inner sep=0pt,minimum height=4.4cm,minimum width=16cm, label={[xshift=-7.3cm,yshift=-0.6cm]\small \itshape disk}, draw] {};

\node (root) at (5.5cm,3.55cm) [rootnode, draw] {};
\node (rootlist) [below right = 0.55cm and 0.4cm of root.north west, anchor=north west, list, draw] {};
\node (rootsignal) [below right = 1.47cm and 0.4cm of root.north west, anchor=north west, signalbuffer, draw] {};
\node (rootboundary) [below right = 1.47cm and 3.2cm of root.north west, anchor=north west, value, label={[xshift=-0.3cm, label distance=-2pt]\small \itshape other parts}, draw] {};

\node (child1) at (0cm,0) [treenode, draw] {};
\node (child1todo) [below right = 0.61cm and 0.4cm of child1.north west, anchor=north west, todobuffer, draw] {};
\node (child1list) [below right = 1.52cm and 0.4cm of child1.north west, anchor=north west, list, draw] {};
\node (child1signal) [below right = 2.43cm and 0.4cm of child1.north west, anchor=north west, signalbuffer, draw] {};
\node (child1bloom) [below right = 0.61cm and 2.4cm of child1.north west, anchor=north west, bloom, draw] {};
\node (child1boundary) [below right = 0.61cm and 3.2cm of child1.north west, anchor=north west, value, draw] {};
\node (child1arrow1) [below right = 0.6cm and -3.5cm of child1] {};
\node (child1arrow2) [below right = 0.6cm and -2.5cm of child1] {};
\node (child1arrow3) [below right = 0.25cm and -2cm of child1] {......};
\node (child1arrow4) [below right = 0.6cm and -0.5cm of child1] {};

\node (child2) at (4.8cm,0) [treenode, draw] {};
\node (child2todo) [below right = 0.61cm and 0.4cm of child2.north west, anchor=north west, todobuffer, draw] {};
\node (child2list) [below right = 1.52cm and 0.4cm of child2.north west, anchor=north west, list, draw] {};
\node (child2signal) [below right = 2.43cm and 0.4cm of child2.north west, anchor=north west, signalbuffer, draw] {};
\node (child2bloom) [below right = 0.61cm and 2.4cm of child2.north west, anchor=north west, bloom, draw] {};
\node (child2boundary) [below right = 0.61cm and 3.2cm of child2.north west, anchor=north west, value, draw] {};
\node (child2arrow1) [below right = 0.6cm and -3.5cm of child2] {};
\node (child2arrow2) [below right = 0.6cm and -2.5cm of child2] {};
\node (child2arrow3) [below right = 0.25cm and -2cm of child2] {......};
\node (child2arrow4) [below right = 0.6cm and -0.5cm of child2] {};

\node (child3) at (8cm,0cm) [minimum height=3.1cm, minimum width=2cm, label={[xshift=0cm, yshift=-2cm]\small \itshape......}]{};

\node (child4) at (11.2cm,0) [treenode, draw] {};
\node (child4todo) [below right = 0.61cm and 0.4cm of child4.north west, anchor=north west, todobuffer, draw] {};
\node (child4list) [below right = 1.52cm and 0.4cm of child4.north west, anchor=north west, list, draw] {};
\node (child4signal) [below right = 2.43cm and 0.4cm of child4.north west, anchor=north west, signalbuffer, draw] {};
\node (child4bloom) [below right = 0.61cm and 2.4cm of child4.north west, anchor=north west, bloom, draw] {};
\node (child4boundary) [below right = 0.61cm and 3.2cm of child4.north west, anchor=north west, value, draw] {};
\node (child4arrow1) [below right = 0.6cm and -3.5cm of child4] {};
\node (child4arrow2) [below right = 0.6cm and -2.5cm of child4] {};
\node (child4arrow3) [below right = 0.25cm and -2cm of child4] {......};
\node (child4arrow4) [below right = 0.6cm and -0.5cm of child4] {};

\path[-] (root.225) edge (child1.north);
\path[-] (root.290) edge (child3.north);
\path[-] (root.320) edge (child4.north);
\path[-] (child1) edge (child1arrow1)
                edge (child1arrow2)
                edge (child1arrow4);
\path[-] (child2) edge (child2arrow2)
                edge (child2arrow4);
\path[-] (child4) edge (child4arrow1)
                edge (child4arrow2)
                edge (child4arrow4);
\draw[<-, ultra thick] (root) -- (child2.83);
\draw[<-, ultra thick] (child2) -- (child2arrow1);
\path [->, ultra thick] (child2.238)edge[out=175, in=215, looseness=2.3] (child2signal.west);
\path [->, ultra thick] (child2signal.west) edge (child2signal.east);
\path [->, ultra thick] (child2signal.east) edge[out=50, in=215, looseness=2.15] (child2list.west);
\path [->, ultra thick] (child2list.west) edge (child2list.east);
\path [->, ultra thick] (child2list.east) edge[out=40, in=215, looseness=1.35] (child2todo.west);
\path [->, ultra thick] (child2todo.west) edge (child2todo.east);
\path [->, ultra thick] (child2todo.east)edge[out=10, in=250, looseness=1.5] (child2.83);
\end{tikzpicture}
\caption{An illustration of the data structure. The arrows show the time order inside one node as defined in Definition \ref{def:time_order}.} \label{fig}
\end{figure}
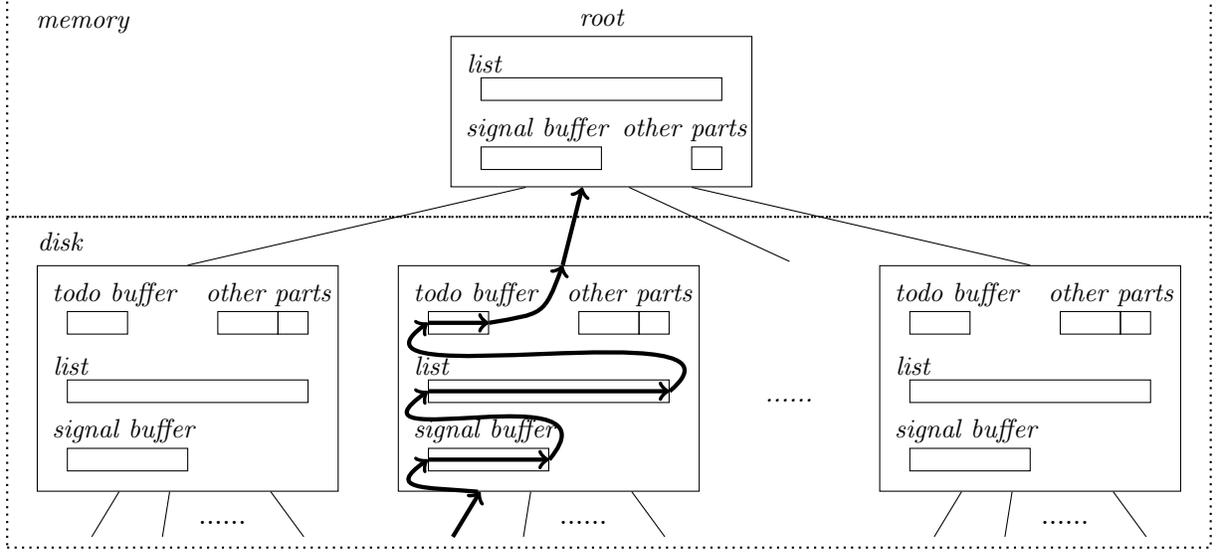

Now we are ready to describe our invariants. For each key $k$ we keep a set of \emph{marked nodes} $M_k$ that stores the nodes where a false positive error has occurred. Initially $M_k=\emptyset$ for every $k$. Some invariants are only defined for the nodes that are not marked. This is because when a Bloom filter replacement makes a false positive error, the restrictions of some invariants may be temporarily violated and they are restored in later procedures.  The notion of marked nodes is only used for analysis and we do not explicitly compute or store the marked nodes in our data structure.

Invariant 1 and 2 states that each key is uniquely stored in the tree with exception of the marked nodes. If there are multiple copies of a key, then there must exist Delete signals that will delete the redundant copies. 
We say that a finite Actual$(k,v)$ is deleted if there exist Delete signals in the signal buffers above it, and we say a finite Actual$(k,v)$ is unique if all the other finite Actual entries and Update signals below it are deleted by Delete signals in the signal buffers below it.
Invariant 1 and 2 essentially state that a finite Actual$(k,v)$ is either unique or deleted if $v$ is not marked.

\begin{Invariant}{1}[Uniqueness] $\forall k$, for any two nodes $u,v$ where $u$ is an ancestor of $v$ and $u\notin M_k$, if both Actual$(k,u)$ and Actual$(k,v)$ are finite, then there must exist a Delete$(k)$ signal in a signal buffer that is between the list of $u$ and the list of $v$ according to the time order.
\end{Invariant}

\begin{Invariant}{2}[Uniqueness] 
$\forall k$, for any node $v$ where $v\notin M_k$ and Actual$(k,v)$ is finite and for any signal buffer \emph{buf} that is lower than the list of $v$, if \emph{buf} contains an Update$(k,p)$ signal, 
then there must exist a Delete$(k)$ signal after the Update signal in \emph{buf} or in a signal buffer between \emph{buf} and the list of $v$.
\end{Invariant}

Invariant 3 states that the final priority of each key up to the root is the correct priority that the tree should store. 
\begin{Invariant}{3}[Correctness]
Let $S$ be the dynamic set that the priority queue maintains. $\forall k$, if there exists an entry $(k,p)$ in $S$, then Final$(k, root)=p$. Otherwise Final$(k, root)=+\infty$.
\end{Invariant}

Invariant 4 and 5 state that heap order is maintained in the tree. Together they state that the boundary values mark the boundaries of the ranges of the priorities of the actual entries in each node. Invariant 5 allows exception for the nodes in $M_k$.
\begin{Invariant}{4}[Heap Order] 
$\forall k$, for any two nodes $u,v$ where $u$ is an ancestor of $v$, we have Boundary$(v)\geq$ Boundary$(u)$, and Actual$(k,v)\geq$ Boundary$(u)$, and $p\geq$ Boundary$(u)$ for any Update$(k,p)$ signal in the signal buffer of $u$,
\end{Invariant}

\begin{Invariant}{5}[Heap Order]
$\forall k$, for any node $v$ where $v\notin M_k$ and Actual$(k,v)$ is finite, we have Actual$(k,v)\leq$ Boundary$(v)$.
\end{Invariant}

By definition of Final$(k,v)$, Final$(k,v)$ can only be updated by Actual$(k,v)$ or Update signals of key $k$ in the signal buffer of $v$, therefore we have the following corollary of Invariant 4:
\begin{noncorollary}
$\forall k$, for any two nodes $u,v$ where $u$ is an ancestor of $v$, we have Final$(k,v) \geq$ Boundary$(u)$. If Actual$(k,u)=+\infty$, we also have Final$(k,u) \geq$ Boundary$(u)$.
\end{noncorollary}

Finally we also maintain two invariants for the todo buffers and the marked nodes.
\begin{Invariant}{6}[Todo Buffers]
$\forall k$, the todo buffer of each node can only store one Delete$(k)$ signal, or store one Update$(k,p)$ signal, or store one Delete$(k)$ signal followed by one Update$(k,p)$ signal.
\end{Invariant}

\begin{Invariant}{7}[Marked Nodes]
For any key $k$ and any node $v\in M_k$, the todo buffer of $v$ stores only one signal with key $k$ which is an Update$(k,p)$ signal where $p>$ Boundary$(k, v)$. The list of $v$ contains no entry with key $k$.
\end{Invariant}

\section{Operations}
We say a buffer or a list overflows when it has more than its maximum number of signals or entries.
To support the basic operations, we introduce the following auxiliary procedures:
\begin{itemize}
\item PushSignal$(v)$: When the signal buffer of $v$ overflows, this procedure is called to push down the signals to the children of $v$. It uses the Bloom filter replacement and the todo buffer of a child to decide if a signal should be applied to the actual entries of that child. If so, the signal should go to the todo buffer of that child. Otherwise the signal should go to the signal buffer of that child.
In our analysis, if an Update$(k,p)$ signal is wrongly put into the todo buffer of a child due to a false positive error, we put that child into $M_k$.
\item ApplyTodo$(v)$: When the todo buffer of $v$ overflows, this procedure is called to apply all the signals in the todo buffer to the entries in the list. After this procedure, the list of $v$ stores the actual entries of $v$.
In our analysis, after this procedure $v\notin M_k$ for any $k$.
\item EmptyList$(v)$: Push all excess entries in the list of $v$ to its children to restore the list to $tB$ entries. This procedure is called when some Update signals insert too many entries in a node.
\item FillUp$(v)$: Move up $tB$ entries from the children of $v$ to $v$ when the list of $v$ is empty. This procedure is called when some Delete signals remove too many entries from a node.
\end{itemize}

Now we describe the operations and the procedures in details.

\subsection{Delete$(k)$ Operation and ExtractMin$()$ Operation}
If there exists an entry with key $k$ in the root, the Delete$(k)$ Operation deletes it, and FillUp(root) is called if the deletion causes the list of the root to become empty. Otherwise a signal Delete$(k)$ is put into the signal buffer of the root, and PushSignal(root) is called if the signal buffer of the root overflows.

ExtractMin$()$ first returns the entry $(k,p)$ with the smallest priority among the list of the root, and then calls Delete$(k)$.

Since the root of the tree is stored in the memory, these two operations require no I/Os if we do not consider the calls to other procedures.

\paragraph{Analysis.} Invariant 1 and 2 hold because these two operations may only change Actual$(k,root)$ to $+\infty$ and they do not affect other actual priorities or create new Update signals. After these two operations Final$(k,root)$ becomes $+\infty$ and Invariant 3 also holds. Finally there are no changes to any todo buffer or any boundary values or any marked nodes, so Invariants 4, 5, 6 and 7 still hold.

\subsection{Update$(k, p)$ Operation}
If there exists an entry $(k, p')$ in the list of the root, it is updated to $(k,\min(p, p'))$. The order of the root list is adjusted to maintain the descending order according to priority.

If key $k$ is not present in the list of the root and $p>$ Boundary(root), an Update$(k, p)$ signal is put into the signal buffer of the root, and PushSignal(root) is called if the signal buffer overflows. 

Finally if $p\leq$ Boundary(root), in order to maintain the heap order, an entry $(k,p)$ is inserted into the list of the root. We then need to delete any other entry with key $k$ in the tree, so a Delete$(k)$ signal is put into the signal buffer of the root and PushSignal(root) is called if the signal buffer overflows. Inserting the entry $(k,p)$ may cause the list of the root to overflow. If this happens, EmptyList(root) is called to restore the list of the root to $tB$ entries. 

The Update$(k,p)$ operation also requires no I/Os if we do not consider the calls to other procedures.

\paragraph{Analysis.} In the first case where there already exists an entry $(k, p')$ in the list of the root, we simply update this entry and Actual$(k,\text{root})=(k,\min(p,p'))$, all the invariants hold trivially. 

In the second case where a signal Update$(k, p)$ is inserted into the signal buffer of the root, since we do not change any actual priority or any boundary value and we do not change the signals in any todo buffer or change any marked nodes, Invariant 1, 5, 6 and 7 hold. Since the only list above the new Update signal is the list of the root and in this case Actual$(k,\text{root})=+\infty$, Invariant 2 also holds. The new value of Final$(k,\text{root})$ is $p$ only if $p$ is smaller than the old value of Final$(k,\text{root})$ and this is the desired behavior of an Update signal, so Invariant 3 holds. Invariant 4 holds because in this case $p>$ Boundary(root).

In the final case where a new entry $(k,p)$ is inserted into the list of the root, since a new Delete signal is inserted into the signal buffer of the root, Invariant 1 and 2 hold. Since the old value of Actual$(k,\text{root})$ is $+\infty$ and using the corollary of Invariant 4 we know the old Final$(k,\text{root})$ value satisfies that Final$(k,\text{root})\geq$  Boundary(root). The new value of Final$(k,\text{root})$ is $p$ and $p\leq$ Boundary(root), so it is the correct priority for $k$ and Invariant 3 also holds. Invariant 4 and 5 hold because $p\leq$ Boundary(root). Finally Invariant 6 and 7 hold since no todo buffer and no marked node is changed.

\subsection{EmptyList$(v)$ Procedure}
This procedure is called when the list of $v$ has more than $2tB$ entries and the excess entries need to be put down to the children of $v$ to restore the list of $v$ to $tB$ entries. We first call ApplyTodo$(v)$ to empty the todo buffer, and then call PushSignal$(v)$ to empty the signal buffer. All components of $v$ and the todo buffers and the Bloom filter replacements of all the children of $v$ are loaded into memory. Loading them costs $O(t)$ I/Os in total.

When putting an entry $(k,p)$ to a child $c$, using Invariant 1 and the fact that the signal buffer of $v$ is empty, we know that Actual$(k,c)=+\infty$. Invariant 6 ensures that the todo buffer of $c$ contains at most one Delete signal with key $k$. The Delete signal is changed to an Update$(k,p)$ signal if it exists. Otherwise $(k,p)$ should be appended to the end of the list of $c$. We make one I/O when we have accumulated $B$ such entries. $k$ is inserted into the Bloom filter replacement of $c$ and deleted from the Bloom filter replacement of $v$. The insertion of entries may cause the list of $c$ to overflow. If this happens, EmptyList$(c)$ is called recursively. After this procedure, we set Boundary$(v)$ to be the maximum priority of the remaining entries in the list of $v$.

\paragraph{Analysis.} After ApplyTodo$(v)$ is called at the beginning of the procedure, $v\notin M_k$ for any $k$. Consider each entry $(k,p)$ that is moved from the list of $v$ to the child $c$. Before this procedure we have Actual$(k,v)=p$ and Actual$(k,c)=+\infty$, so using Invariant 7 we have $c\notin M_k$. After this procedure we have Actual$(k,v)=+\infty$ and Actual$(k,c)=p$. Since the signal buffer of $v$ is empty, the Delete signals of Invariant 1 and 2 that was used to make sure Actual$(k,v)$ is either unique or deleted can now be used for Actual$(k,c)$. So Invariant 1 and 2 still hold. Invariant 3 holds because after this procedure Final$(k,v)$ is still $p$. 
Boundary$(v)$ is decreased after this procedure. Using Invariant 5, the priorities of the remaining entries in the list of $v$ are between the new Boundary$(v)$ and the boundary values of the ancestors of $v$, so we still have Boundary$(v)\leq $ Boundary$(u)$ for any ancestor $u$ of $v$. Since the entries with largest priorities are moved down, we have $p \geq$ Boundary$(v)$. So Invariant 4 still holds. Before $(k,p)$ is moved down we have $p=$ Actual$(k,v)\leq$ Boundary$(v)\leq$ Boundary$(c)$, so after this procedure Actual$(k,c)=p\leq$ Boundary$(c)$. Invariant 5 still holds. The todo buffer of $c$ may be changed to contain only one Update$(k,p)$ signal, so Invariant 6 still holds. Finally this procedure itself does not change any marked node if we do not consider the calls to other procedures, so Invariant 7 still holds.

\subsection{FillUp$(v)$ Procedure}
When the list of $v$ is empty, we call this procedure to find the $tB$ entries with smallest priorities from the lists and the todo buffers of all children of $v$ and move them up to $v$. We first call ApplyTodo$(v)$ to empty the todo buffer, and then call PushSignal$(v)$ to empty the signal buffer. 

We load the $B$ entries with smallest priorities from the end of the list of each child and also delete them from the Bloom filter replacements of that child.  
When the loaded entries of a particular child have all been moved up, another $B$ entries are loaded from the list of that child. Removing entries from the list of a child $c$ may cause the list of $c$ to be empty. If this happens, FillUp$(c)$ is called recursively.
All components of $v$ and all todo buffers and all Bloom filter replacements of the children of $v$ are also loaded into memory. Loading them costs $O(t)$ I/Os in total.

We repeatedly find the entry $(k,p)$ with the smallest priority among all the loaded entries and all the Update signals in the todo buffers. When $(k,p)$ is found as an Update signal in the todo buffer of a child $c$, using Invariant 6 we know that there exists no Delete$(k)$ signal following it, so Actual$(k,c)=p$. The Update signal is changed to a Delete$(k)$ signal and any other signal with key k is removed from the todo buffer. 
When $(k,p)$ is found among the loaded entries from a child $c$, if there exists a Delete$(k)$ signal in the todo buffer of $c$, we simply delete the entry. Otherwise we move up the entry and remove the possible Update$(k,p')$ signal from the todo buffer of $c$.The entry $(k,p)$ is then inserted into the list of $v$ and $k$ is inserted into the Bloom filter replacement of $v$.

The above process is repeated until the list of $v$ has $tB$ entries. 
Finally we put the remaining loaded entries back to the list of the children and also insert their keys back to the Bloom filter replacements. Boundary$(v)$ is updated to be the maximum priority of the entries in the list of $v$.

\paragraph{Analysis.} 
After ApplyTodo$(v)$ is called at the beginning of the procedure, $v\notin M_k$ for any $k$. During the procedure no Update signal in the todo buffer of $c$ with larger priority than Boundary$(c)$ is considered, so using Invariant 7 we also have $c\notin M_k$ for each $(k,p)$ that is moved up from $c$. 

Consider each new $(k,p)$ entry in the list of $v$, this entry could be moved up from an Update$(k,p)$ signal in the todo buffer of $c$ or from an entry in the list of $c$. In both cases Actual$(k,c)=p$ and Actual$(k,v)=+\infty$ before this procedure, and Actual$(k,c)=+\infty$ and Actual$(k,v)=p$ after this procedure, Invariant 1 and 2 hold because the Delete signals that were used to make sure Actual$(k,c)$ is either unique or deleted can now be used for Actual$(k,v)$. Invariant 3 holds because after this procedure Final$(k,v)$ is still $p$. The new Boundary$(v)$ equals the maximum priority of the moved up entries. Using Invariant 5, we know that for any remaining actual entry $(k',p')$ in $c$, $p'\leq$ Boundary$(c)$, so the new Boundary$(v) \leq p'\leq$ Boundary$(c)$, so Invariant 4 still holds. Invariant 5 also holds because Boundary$(v)$ can only be increased. After this procedure the todo buffer of $c$ contains no signal with key $k$ or only one Delete$(k)$ signal, so Invariant 6 still holds. Finally this procedure itself does not change any marked node if we do not consider the calls to other procedures, so Invariant 7 still holds.

\subsection{PushSignal$(v)$ Procedure}
This procedure pushes the signals in the signal buffer of $v$ down to its children. All components of $v$ and all todo buffers and all Bloom filter replacements of the children of $v$ are first loaded into memory. Loading them costs $O(t)$ I/Os in total.

We first define a method CheckInActual$(k,c)$ which checks if a child $c$ contains an actual entry with key $k$ using only the Bloom filter replacement and the todo buffer of $c$. The method returns true if $k$ is in the Bloom filter replacement or if the todo buffer of $c$ contains an Update signal with key $k$. Using Invariant 6 we know that the Update signal in the todo buffer is not followed by any Delete signal. Otherwise the method returns false. Note that when CheckInActual$(k, c)$ returns true, it is guaranteed that Actual$(k,c)\neq +\infty$ with at most $\varepsilon$ probability of a false positive error. When CheckInActual$(k,c)$ returns false, we can safely assert that Actual$(k,c)= +\infty$ because the Bloom filter replacement makes no false negative errors.

We push down each signal in the signal buffer of $v$ from start to end. Note that we make one I/O when we have accumulated $B$ signals to be put into the same buffer. Consider a signal that should go to a child $c$. If $c$ is a leaf, then the signal always goes into the todo buffer of $c$.

Now consider when $c$ is an internal node. For a Delete$(k)$ signal, we first use CheckInActual$(k,c)$ to check if $v$ contains an actual entry with key $k$. Delete$(k)$ is put into the todo buffer of $c$ only if CheckInActual$(k,c)$ returns true, and in this case all other signals with key $k$ are removed from the todo buffer. Since a false positive error may occur, we always put a Delete$(k)$ signal into the signal buffer of $c$. In our analysis, if $c\in M_k$, after pushing down the Delete signal we remove $c$ from $M_k$.
 
For an Update$(k, p)$ signal, if $p\leq$ Boundary$(c)$, then no matter whether $k$ is present in the actual entries of $c$ or not, the Update signal should go into the todo buffer of $c$. If there already exists another Update$(k,p')$ signal in the todo buffer of $c$, we simply update its priority to $\min(p,p')$. We also put a Delete$(k)$ signal into the signal buffer of $c$ to delete any other entry with key $k$ in the tree. In our analysis, if $c\in M_k$, after pushing down the Update signal we remove $c$ from $M_k$. 

If $p>$ Boundary$(c)$, we again use CheckInActual$(k,c)$ to check if Actual$(k,v)\neq +\infty$. If it returns true, the Update signal should go into the todo buffer of $c$. If there already exists another Update$(k,p')$ signal in the todo buffer of $c$, we simply update its priority to $\min(p,p')$. In the analysis we put $c$ into $M_k$ if CheckInActual$(k,c)$ makes a false positive error and there is no entry with key $k$ in the list of $c$.  
Finally if $p>$ Boundary$(c)$ and CheckInActual$(k,c)$ returns false, there can be no error and Update$(k, p)$ is put into the signal buffer of $c$.

During this procedure, PushSignal$(c)$ is called whenever the signal buffer of $c$ overflows and ApplyTodo$(c)$ is called whenever the todo buffer of $c$ overflows.

\paragraph{Analysis.} We prove that after pushing down each signal, all the invariants hold. First consider a Delete$(k)$ signal. After pushing down this signal, Actual$(k,c) = +\infty$ and the last signal in the signal buffer of $c$ is a Delete$(k)$ signal. This Delete$(k)$ signal effectively deletes any actual entry with key $k$ or any Update$(k,p)$ signal under $v$. So Invariant 1 and 2 still hold. 
Since Final$(k,c)=+\infty$ after the Delete signal is pushed down, Invariant 3 still holds.
Since no boundary value is changed and the only change to any actual priority is to change Actual$(k,c)$ to $+\infty$, Invariant 4 and 5 still hold. The todo buffer of $c$ is either unchanged or updated to contain only one Delete signal with key $k$, so Invariant 6 still holds. Finally, Invariant 7 holds since we do not mark new nodes and $c\notin M_k$ after pushing down the Delete signal.

Next consider pushing down an Update$(k,p)$ signal to the child $c$. Call this signal $\sigma$. If $p\leq$ Boundary$(c)$, after putting $\sigma$ into the todo buffer of $c$, Actual$(k,c)$ is finite. There is a newly placed Delete$(k)$ signal in the signal buffer of $c$ right after it. And the Delete signal of Invariant 2 that deletes $\sigma$ when it was in the signal buffer of $v$ can still delete the new actual entry. So Invariant 1 still holds. There is no new Update signals in any signal buffer, so Invariant 2 still holds. Final$(k,v)$ remains the same as before, so Invariant 3 still holds. 
No boundary value is changed and $p>$ Boundary$(v)$ using Invariant 4, so Actual$(k,c)\geq$ Boundary$(v)$ and Invariant 4 still holds. Invariant 5 holds because $p\leq$ Boundary$(c)$. Since we remove any other existing Update signal with key $k$ from the todo buffer of $c$, Invariant 6 still holds. Finally, Invariant 7 holds since we do not mark new nodes and $c\notin M_k$ after pushing down $\sigma$.

If $p>$ Boundary$(c)$ and CheckInActual$(k,c)$ returns false, then we are sure that Actual$(k,c)=+\infty$ before and after $\sigma$ is pushed down. If there exists a Delete$(k)$ signal that was used to effectively delete $\sigma$ before it is pushed down, the Delete signal can still effectively delete $\sigma$ after it is pushed down to the signal buffer of $c$, so Invariant 2 still holds.
Other invariants hold since there is no changes to any actual priority or any boundary value or any todo buffer or any marked node.

If $p>$ Boundary$(c)$ and CheckInActual$(k,c)$ returns true, if the Bloom filter replacement makes no false positive error, if $c$ is already marked, then the Update signal in the todo buffer of $c$ is updated and all Invariants still hold. If $c$ is not marked, then using Invariant 5 we know that before pushing down $\sigma$ we have Actual$(k, c)\leq$ Boundary$(c) < p$, so after putting $\sigma$ into the todo buffer of $c$, Actual$(k,c)$ remains the same. So all Invariants still hold. 

If the Bloom filter replacement makes a false positive error, then no entry or signal with key $k$ was stored in the list or the todo buffer of $c$ and $\sigma$ is wrongly placed into the todo buffer of $c$. In our analysis, $c$ is put into $M_k$. Actual$(k, c) = +\infty$ before $\sigma$ is pushed down, and Actual$(k, c) = p$ after $\sigma$ is pushed down. Since $c\in M_k$, we no longer require that there is a Delete signal that follows the list of $c$. And the Delete signal of Invariant 2 that deletes $\sigma$ when it was in the signal buffer of $v$ can still delete the new actual entry. Hence Invariant 1 still holds. 
Invariant 2, 5, 6 and 7 hold trivially.
Before $\sigma$ is put down, let $p'=$ Final$(k, c)$, and when computing Final$(k, v)$, when reaching the signal just after $\sigma$ in the signal buffer of $v$, the value of Final$(k, v)$ is $\min(p, p')$. After $\sigma$ is put down, Actual$(k, c)=p$ and Final$(k,c) = \min(p,p')$. So when reaching the signal just after $\sigma$ in the signal buffer of $v$, Final$(k, v)$ remains to be $\min(p, p')$. Hence Invariant 3 still holds. 
Since no boundary value is changed and now Actual$(k, c)=p>$ Boundary$(v)$, so Invariant 4 still holds.

\subsection{ApplyTodo$(v)$ Procedure}
If $v$ is the root, $v$ has no todo buffer and we update Boundary$(v)$ to be the maximum priority in the list. If $v$ is an internal node whose todo buffer is empty, the procedure returns without doing anything. Otherwise, applying the signals of the todo buffer is similar to what is done to the root by the Delete and Update operations. All components of $v$ are first loaded into memory. Loading them costs $O(t)$ I/Os in total.

Consider each signal in the todo buffer of $v$ from start to end. If the signal is Delete$(k)$, then any existing entry with key $k$ is deleted from the list of $v$. 

If the signal is Update$(k,p)$, we first check if the Bloom filter replacement has made a false positive error in the PushSignal procedure and this signal is put here by mistake. Using Invariant 7 we know that a false positive error occurs only when $p>$ Boundary$(v)$ and no entry with key $k$ is present in $v$. In this case the todo buffer of $v$ only has this one signal with key $k$. The Update$(k,p)$ signal is put back into the signal buffer of $v$ and in our analysis we remove $v$ from $M_k$. PushSignal$(v)$ is called if the signal buffer of $v$ overflows.
If no error occurs and there exists an entry $(k,p')$ in $v$, this entry is updated to $(k,\min(p,p'))$. Finally if no entry with key $k$ is present in the entries of $v$ and $p\leq$ Boundary$(v)$, an entry $(k,p)$ is inserted into the list of $v$.

In the end the list of $v$ is sorted in descending order according to the priorities and we rebuild the Bloom filter replacement of $v$. Boundary$(v)$ is updated to be the maximum priority in the list of $v$. EmptyList$(v)$ is called if the list of $v$ overflows and FillUp$(v)$ is called if the list of $v$ is empty. 

\paragraph{Analysis.} After the ApplyTodo procedure the todo buffer of $v$ is emptied and $v$ is removed from all $M_k$, so Invariant 6 and 7 hold. In the end of the procedure Boundary$(v)$ is updated to be the maximum priority of entries stored in the list of $v$, so Invariant 5 still holds.

For any key $k$, if the Bloom filter replacement did not make any false positive error for signals with $k$, then after this procedure Actual$(k,v)$ remains the same and no new signal is put into the signal buffer of $v$, so invariants 1, 2, 3 still hold. Since only entries with smaller priority than the old Boundary$(v)$ are put into the list of $v$, Boundary$(v)$ is only decreased and Invariant 4 still holds. 

Now consider any key $k$ where $v\in M_k$. After the ApplyTodo$(v)$ procedure, Actual$(k,v)$ becomes $+\infty$ and Invariant 1 still applies to old pairs of nodes. An Update$(k,p)$ is put into the signal buffer of $v$, since before this procedure, Actual$(k,v)=p$, for any ancestor $u$ of $v$ such that $u\notin M_k$, using Invariant 1 we know that if Actual$(k,u)\neq +\infty$, then there exist Delete signals in the signal buffers between the list of $v$ and the list of $u$. So after this procedure, there exist Delete signals in the signal buffers between the new Update$(k,p)$ signal and the list of $u$. So Invariant 2 still holds. After this procedure, Final$(k,v)$ remains to be $\min(p,$ Final$(k,c))$ where $c$ is a child of $v$, so Invariant 3 still holds. Since Boundary$(v)$ is only decreased after this procedure and $p>$ old Boundary$(v)$ using Invariant 7, so Invariant 4 still holds.

\section{Analysis}\label{sec:analysis}
In this section we provide a full analysis of our data structure. We first prove the properties of the Bloom filter replacement that we stated in Section \ref{sec:data}. Then we use the invariants to prove the correctness of the data structure. Finally we use a credit argument to show that the I/O complexity of the data structure is $O(\frac{1}{B}\log \frac{N}{B}/\log\log N)$. Together, Theorem \ref{thm:correctness} and Theorem \ref{thm:IO} imply our main result of Theorem \ref{thm:main}.

\subsection{Bloom Filter Replacement}\label{sec:bfr}
We first state the following theorem of \cite{MR2298337}.
\begin{theorem}\label{theorem:bfr}
Let $n\in \mathbb{N}$ and $\varepsilon> 0$, and let $w$ be the word size. We can maintain a data structure (Bloom filter replacement) for a dynamic multiset $S$ of size at most n, whose elements are from $\{0,1\}^w$, such that:
\begin{itemize}
\item Inserting in S and deleting from S can be done in amortized expected constant time. 
\item The data structure can only make false positive errors with probability at most $\varepsilon$, i.e., if $x\in S$, the data structure always answers correctly and if $x\notin S$, the data structure may wrongly answer that $x$ is in $S$ with probability at most $\varepsilon$.
\item The space usage is at most $\left(1 + o\left(1\right)\right) n \log \frac{1}{\varepsilon} + O\left(n + w\right)$ bits.
\end{itemize}
\end{theorem}

The following corollary follows directly and this corollary shows the correctness of the properties of Bloom filter replacement that we stated in Section \ref{sec:data}:
\begin{corollary}
A Bloom filter replacement is stored in each node to record the keys of the entries of the node and it can support both insertion and deletion of keys. Each Bloom filter replacement can only make false positive errors with probability at most $\varepsilon=\frac{1}{\log^3 N}$ and can be stored in $O(B)$ words.
\end{corollary}
\begin{proof}
Each Bloom filter replacement stores at most $2tB$ keys where $t = \log^{0.01}N$. From Theorem \ref{theorem:bfr} it follows that a Bloom filter replacement can support both insertion and deletion of keys correctly. Each Bloom filter replacement can be stored in $\left(1 + o\left(1\right)\right) 2tB \log \frac{1}{\varepsilon} + O\left(2tB + w\right)=O(2tB \log \frac{1}{\varepsilon}  + w) = O(\log^{0.01}N  \log\log N \cdot B+ w)$ bits. Since the word size $w = \log N$, each Bloom filter replacement can be stored in $O(B)$ words.
\end{proof}

\subsection{Correctness}\label{sec:correctness}
We use the invariants to prove the correctness of the data structure: 
\begin{theorem}[Correctness]\label{thm:correctness}
The ExtractMin operation always returns the correct entry.
\end{theorem}

\begin{proof}
The ExtractMin operation returns the entry $(k,p)$ with the smallest priority from the list of the root. We need to prove two statements: $p$ is indeed the correct priority of $k$. $(k,p)$ has the smallest priority among all entries stored in the data structure.

Since the root has no todo buffer, Actual$(k,\text{root})=p$. Using Invariant 1 and 2, we know that for any node $u$ where Actual$(k,u)\neq +\infty$ and any Update$(k,p')$ signal in a signal buffer, there exists a Delete$(k)$ signal in the signal buffers above them. Therefore Final$(k,\text{root})=$ Actual$(k,\text{root})=p$. Using Invariant 3, we know that $p$ is indeed the correct priority of $k$. 

Now suppose an entry $(k',p')$ is stored in the priority queue and $p'<p$. It cannot be stored in the list of the root because if so ExtractMin$()$ would return $(k',p')$. Hence Actual$(k',\text{root})=+\infty$. Using invariant 4 and 5 we also have that no Update$(k',p')$ signal can be stored in the signal buffer of the root because this would imply $p'\geq $ Boundary$(\text{root})\geq p$. Let $c$ be the child of the root that is on the path to Leaf$(k')$. Using Invariant 3 we have Final$(k',\text{root})=p'$, so we must have Final$(k', c)=p'$. Using the corollary of Invariant 4 we know that Final$(k', c)\geq$ Boundary$(\text{root})\geq p$, which leads to contradiction. Hence $p$ is the smallest priority among all entries stored in the data structure. Therefore $(k,p)$ is the correct entry to be returned by the ExtractMin operation.
\end{proof}

\subsection{I/O-complexity} \label{sec:IO}
We first make an useful observation, and then we use the observation to prove the main theorem.
\begin{observation} The height of the tree is $h=O(\log_t \frac{N}{B})=O(\log \frac{N}{B}/\log\log N)$. Each Delete, ExtractMin and Update operation requires no I/Os and each PushSignal, ApplyTodo, EmptyList and FillUp procedure takes $O(t)$ I/Os when excluding the recursive calls to other procedures.
\end{observation}

\begin{theorem}[I/O complexity]\label{thm:IO}
The expected amortized cost for each Delete, Update, ExtractMin operation is $O(\frac{h}{B}+\frac{t}{B})=O(\frac{1}{B}\log \frac{N}{B}/\log\log N)$ I/Os. 
\end{theorem}

\begin{proof}
A Delete operation and an ExtractMin operation each generates at most one Delete signal. An Update operation generates at most one Update signal and one Delete signal. So each operation can generate at most 2 signals to be propagated down. Thus the expected amortized cost of an operation is at most 2 times the expected amortized cost of a signal. 
We prove the theorem using a standard credit argument. 
Each node charges the signals that go into its buffers some credits, and when performing a procedure on that node, it has accumulated enough credits to pay for the required I/Os. In the end, we sum up the total number of credits that a signal is charged and this is the amortized cost of that signal. Now we consider the procedures one by one.

\noindent \textbf{PushSignal:} Each signal is charged $\frac{1}{B}$ credits by a node $v$ when it goes into the signal buffer of $v$. Since each signal goes into at most one signal buffer at each level, it is charged at most $\frac{h}{B}$ credits. When the procedure PushSignal$(v)$ is called, $v$ has already accumulated $tB\cdot \frac{1}{B}=t$ credits to pay for the $O(t)$ cost of this procedure.

\noindent \textbf{ApplyTodo:} Each signal is charged $\frac{t}{B}$ credits by a node $v$ when it goes into the todo buffer of $v$. When the procedure ApplyTodo$(v)$ is called, $v$ has already accumulated $B\cdot \frac{t}{B}=t$ credits to pay for the $O(t)$ cost of this procedure. Each signal goes into the todo buffer of the correct node that the signal should be applied to, and it may wrongly go into the todo buffer when the Bloom filter replacement of that node makes a false positive error. In expectation, each signal goes into $1+\varepsilon h<2$ todo buffers. So in expectation, it is charged at most $\frac{2t}{B}$ credits.

\noindent \textbf{EmptyList:} Each Update$(k,p)$ signal is charged $\frac{1}{B}$ credits by a node $v$ when the entry $(k,p)$ is inserted into the list of $v$. We also let $v$ repay these credits back to the Update$(k,p)$ signal when the entry $(k,p)$ is moved up by a FillUp procedure. Each Update$(k,p)$ signal is only charged these credits when it moves downward in the tree, so in total it is only charged $\frac{h}{B}$ credits. When the procedure EmptyList$(v)$ is called, $v$ has already accumulated $tB\cdot \frac{1}{B}=t$ credits to pay for the $O(t)$ cost of this procedure. Note that EmptyList first calls ApplyTodo and PushSignal, so it also needs to pay the $O(t)$ cost for ApplyTodo and PushSignal.

\noindent \textbf{FillUp:} When a Delete signal deletes its target entry from a node $v$ of height $h(v)$, it pays $\frac{h(v)}{B}$ credits to $v$. Since $h(v)\leq h$, a Delete signal is charged at most $\frac{h}{B}$ credits. When FillUp$(v)$ is called, the node $v$ pays $\frac{h(v)-1}{B}$ credits to a child for each entry that is moved up from the list of that child. The node $v$ also pays $t$ credits for the cost of this FillUp procedure. Since $tB$ entries are moved up to $v$, in total $v$ pays $tB\cdot \frac{h(v)-1}{B}+t=th(v)$ credits. $v$ has accumulated enough credits because when FillUp$(v)$ is called, $v$ has lost $tB$ entries from its list, and each depletion is paid $\frac{h(v)}{B}$ credits, either by a Delete signal or by the parent of $v$ during a FillUp procedure.
Note that EmptyList first calls ApplyTodo and PushSignal, so it also needs to pay the $O(t)$ cost for ApplyTodo and PushSignal.

Summing up, the expected amortized cost of each signal is $O(\frac{h}{B}+\frac{t}{B})=O(\frac{1}{B}\log \frac{N}{B}/\log\log N)$. Hence the amortized cost of each operation is also $O(\frac{1}{B}\log \frac{N}{B}/\log\log N)$.
\end{proof}

\section{Conclusion and Discussion}
In this paper we proposed an external memory priority queue that supports all operations in expected amortized $O(\frac{1}{B}\log \frac{N}{B}/\log\log N)$ I/Os. There still exists a gap between our data structure and the lower bound of $\Omega(\frac{1}{B}\log B/\log\log N)$. The lower bound has a $\log B$ numerator that is independent of $N$, and it would be exciting if one can design a matching priority queue. 
We also note that the lower bound of \cite{MR3678253} is proved in a setting where the data structure only supports the Delete($k)$ operation, and the priorities in the hard distribution are integers smaller than $\log N$. In the more general setting as that of our data structure, there might exist a tighter lower bound.

\bibliographystyle{abbrv}
\bibliography{ref}
\end{document}